\documentclass[conference]{IEEEtran}
\IEEEoverridecommandlockouts
\usepackage{algorithmic}
\usepackage{graphicx}
\usepackage{textcomp}
\usepackage{xcolor}
\def\BibTeX{{\rm B\kern-.05em{\sc i\kern-.025em b}\kern-.08em
    T\kern-.1667em\lower.7ex\hbox{E}\kern-.125emX}}
\usepackage{bm}

\usepackage{amsfonts}
\usepackage{amsmath}

\usepackage{amssymb}
\usepackage{amsthm} 

\DeclareMathOperator{\Tr}{Tr}
\usepackage{algorithmic}
\usepackage{algorithm}
\usepackage{graphicx}
\usepackage{textcomp}
\usepackage{xcolor}
\usepackage{color}
\usepackage{float}
\usepackage{subfigure}
\PassOptionsToPackage{hyphens}{url}\usepackage{hyperref}

\usepackage{graphicx}  %Required
\usepackage{caption}
 
\DeclareCaptionLabelFormat{lc}{{#1}~#2}
\def\tablename{Table}

\usepackage{pifont} 

\makeatletter

\def\endthebibliography{%
  \def\@noitemerr{\@latex@warning{Empty `thebibliography' environment}}%
  \endlist
}
\makeatother

\usepackage{epsfig,endnotes}

\usepackage{floatpag,enumitem}
\usepackage{cite}
\usepackage{relsize}
 \usepackage{tikz}

\usepackage{multirow}
\usepackage{setspace}
\usepackage{mathrsfs}
\usepackage{bbm}
\usepackage{pifont}
\usepackage{soul}
\usepackage{tabularx}
\usepackage{listings}
\usepackage{graphicx}
\usepackage{booktabs}
\usepackage{array}
\usepackage{fancyhdr}
\usepackage{cases}
\usepackage{url}

 \usepackage{supertabular}
\usepackage{url}
\AtBeginDocument{%
  \providecommand\BibTeX{{%
    Bib\TeX}}}

\newtheorem{assumption}{Assumption}

\newtheorem{lem}{Lemma}

\newtheorem{prop}{Proposition}
\newtheorem*{proposition1.1}{Proposition 1.1}
\newtheorem*{proposition1.2}{Proposition 1.2}
\newtheorem*{proposition1.3}{Proposition 1.3}
\newtheorem*{proposition2.1}{Proposition 2.1}
\newtheorem*{proposition2.2}{Proposition 2.2}
\def\BibTeX{{\rm B\kern-.05em{\sc i\kern-.025em b}\kern-.08em
    T\kern-.1667em\lower.7ex\hbox{E}\kern-.125emX}}

\begin{document}

\title{Optimization for the Metaverse over Mobile Edge Computing with Play to Earn\vspace{-5pt}}

\author{Chang Liu$^1$, Terence Jie Chua$^1$, Jun Zhao$^2$\\
$^1$Graduate College, Nanyang Technological University, Singapore\\
$^2$School of Computer Science \& Engineering, Nanyang Technological University, Singapore\\
E-mail:\{liuc0063@e.ntu.edu.sg, terencej001@e.ntu.edu.sg, junzhao@ntu.edu.sg\}\vspace{-15pt}
}
% \author{\IEEEauthorblockN{Chang Liu}
% \IEEEauthorblockA{
% \textit{Graduate College}\\
% \textit{Nanyang Technological University}\\
% Singapore\\
% liuc0063@e.ntu.edu.sg}
% \and
% \IEEEauthorblockN{Terence Jie Chua}
% \IEEEauthorblockA{
% \textit{Graduate College}\\
% \textit{Nanyang Technological University}\\
% Singapore\\
% terencej001@e.ntu.edu.sg}
% \and
% \IEEEauthorblockN{Jun Zhao}
% \IEEEauthorblockA{
% \textit{School of Computer Science \& Engineering}\\
% \textit{Nanyang Technological University}\\
% Singapore\\
% junzhao@ntu.edu.sg}
% }
\thispagestyle{fancy}
\lhead{This work appears as a full paper in IEEE Conference on Computer Communications (INFOCOM) 2024.}
\cfoot{\thepage}
\renewcommand{\headrulewidth}{0.4pt}
\renewcommand{\footrulewidth}{0pt}

\maketitle
\thispagestyle{fancy}
\lhead{This work appears as a full paper in IEEE Conference on Computer Communications (INFOCOM) 2024.}
\cfoot{\thepage}
\renewcommand{\headrulewidth}{0.4pt}
\renewcommand{\footrulewidth}{0pt}
\pagestyle{fancy}
\lhead{This work appears as a full paper in IEEE Conference on Computer Communications (INFOCOM) 2024.}
\cfoot{\thepage}
\renewcommand{\headrulewidth}{0.4pt}
\renewcommand{\footrulewidth}{0pt}

\begin{abstract}
The concept of the Metaverse has garnered growing interest from both academic and industry circles. 
% Through the integration of mobile edge computing (MEC) and mobile augmented reality technologies, users can join the digital society within the Metaverse and interact with others through avatars. 
The decentralization of both the integrity and security of digital items has spurred the popularity of \textit{play-to-earn} (P2E) games, where players are entitled to earn and own digital assets which they may trade for physical-world currencies. However, these computationally-intensive games are hardly playable on resource-limited mobile devices and the computational tasks have to be offloaded to an edge server. Through mobile edge computing (MEC), users can upload data to the Metaverse Service Provider (MSP) edge servers for computing. 
Nevertheless, there is a trade-off between user-perceived in-game latency and user visual experience. The downlink transmission of lower-resolution videos lowers user-perceived latency while lowering the visual fidelity and consequently, earnings of users. 
In this paper, we design a method to enhance the Metaverse-based mobile augmented reality (MAR) in-game user experience.
Specifically, we formulate and solve a multi-objective optimization problem.
Given the inherent NP-hardness of the problem, we present a low-complexity algorithm to address it, mitigating the trade-off between delay and earnings.
% by controlling the transmit power of users, the downlink video resolution and user-to-edge association strategy.
% The prowess of our proposed method is demonstrated through simulations. 
The experiment results show that our method can effectively balance the user-perceived latency and profitability, thus improving the performance of Metaverse-based MAR systems.
%computation offloading  problem. Previous studies ?  
\end{abstract}

\begin{IEEEkeywords}
Metaverse, mobile edge computing, \textit{play-to-earn}, latency, wireless networks.
\end{IEEEkeywords}


%\maketitle
%\thispagestyle{empty}
% \pagestyle{plain}

\vspace{-7pt}\section{Introduction}\vspace{-3pt}

The Metaverse is a virtual environment on public blockchain technology where users create, share, and own digital assets, interacting through virtual avatars. Each world within the Metaverse has its own digital currency, facilitating economic activities like in-game purchases and virtual land acquisition. These digital tokens hold significant value, motivating people to play games to earn them~\cite{farrington_2021}. With a growing trend towards digital activities and gaming~\cite{grayscale}, the combination of better games, AR/VR advancements, and the desire for asset ownership drives users to participate in \textit{play-to-earn} (P2E) games within the Metaverse. This ownership of digital assets and tokens opens up opportunities for virtual world economies and trade~\cite{business}.

The Metaverse and its extended reality applications such as P2E present computation-intensive challenges for game graphics and gameplay \cite{cai2022compute}, especially on current mobile device computing hardware with limited resources. To ensure a seamless real-to-world experience for users, scenes delivered by VR devices must have a minimal delay (within 20ms) and demand high bitrates up to Gbps for synchronization \cite{yu20226g}. To address this limitation, the convergence of MEC with the Metaverse gives rise to a new generation of MEC-empowered Metaverse~\cite{zhang2017towards}. This integration prioritizes enhancing real-time performance in virtual reality systems, system mobility, and ultra-reliability, leveraging the capabilities of B5G or 6G technologies.

In the Metaverse system integrated with MEC, efficient communication is crucial for a seamless virtual world experience. However, high-resolution graphics and large data transmission can lead to increased latency, affecting players' interactions and experiences in P2E games. Minimizing latency is vital to ensure player satisfaction and profitability. Lowering video resolution can reduce data size and latency, but it compromises graphic quality, impacting gameplay and token earnings. Balancing resolution and token earnings becomes a trade-off. Additionally, optimizing user-server association is necessary to address workload imbalance among edge servers, which can cause service latency issues. Developing an effective user-server association strategy is essential for enhancing Metaverse service performance and reducing latency.

In this paper, we identify a problem where MSP edge servers have limited resources for computation and communication, limiting their ability to support a rapidly growing player base. In the absence of resource allocation regulation, players will always expect the best possible game environment in terms of quality and connectivity. However, this may not be feasible for a resource-limited MSP edge server. We present a framework for optimizing communication and computation that strikes a balance between the delay experienced by players and their in-game token earnings. 
This scheme aids the MSP in controlling transmission data and the association of players with MSP edge servers, while also motivating skilled players to earn more game tokens. Our contributions can be summarized as follows:
\begin{itemize}[leftmargin=*]
    \item We introduce the Metaverse system over MEC. We explain the profitability of playing games in the virtual world.
    \item We identify a communication and computation latency problem in mobile augmented reality (MAR) system, and formulated a multi-objective optimization problem which aims to maximize the user utility by controlling the transmit power of users, the resolution of the downlink video and efficiently allocating players to Metaverse edge servers. 
    \item  To solve the multi-objective optimization problem, alternating optimization is utilized. To handle the non-convexity, we transform it into a quadratically constrained quadratic programming (QCQP) problem, and apply semidefinite relaxation (SDR) approach to get an approximation solution. After the optimal user-to-server association is obtained, we optimize the resolution of the downlink video alternately, thus balancing the trade-off between service latency and potential token earnings.
    \item We propose and evaluate the earning functions on real-world datasets. Our proposed algorithm is shown to have benefits in simulation results, as it effectively enhances players' in-game earnings while keeping latency low. This results in improved user utility.

\end{itemize}

The remainder of this paper is organized as follows. In Section \ref{sec-Related-Work}, we discuss related works.
In Section \ref{sec-System}, we introduce the system model for Metaverse-based MEC system. 
Next, we present the formulation and analysis of the problem in Section \ref{sec-analysis}.
In Section \ref{sec-function}, we evaluate the earning function of players by analyzing a real-world dataset.
The simulation results are presented in Section \ref{sec-Simulation}. 
Finally, we conclude the paper in Section \ref{sec-Conclusion}.

\vspace{-4pt}\section{Related works} \vspace{-1pt}\label{sec-Related-Work}
\textbf{Resource Allocation in the Metaverse:} Some existing works consider the resource allocation problem in Metaverse-based systems.
Chu \textit{et al.}~\cite{chu2022metaslicing} introduces the MetaSlicing framework that can efficiently manage and allocate different types of resources by grouping applications into clusters and designing an admission control algorithm.
Han \textit{et al.}~\cite{han2021dynamic} proposes a resource allocation framework for Internet of Things to facilitate the syncing of Metaverse with physical world.
Ng \textit{et al.}~\cite{ng2021unified} propose a framework that uses stochastic optimal resource allocation to minimize the cost to an MSP in the context of the education sector.
However, these works fail to consider the communication between the Metaverse users and the MSP.

In contrast to the abovementioned works, Si \textit{et al.}~\cite{si2022resource} proposes a communication resource allocation algorithm in the Metaverse-based MAR system.
Nevertheless, it only considers one centralized base station to provide Metaverse service. 
In our paper, we consider multiple MSP edge servers and design the user-to-edge association strategy.
Van \textit{et al.}~\cite{van2022edge} jointly considers the integrated model of communications, computing, and storage policy through ultra-reliable and low latency communications (URLLC).
However, they only aim at minimizing the latency in the objective function.
Besides, they fail to consider video resolution as part of the optimization problem.

\textbf{Task-Offloading in Mobile Edge Computing (MEC):} Tasks-offloading in MEC has attracted considerable attention in prior works (e.g., \cite{dai2018joint,tan2017online,jia2016cloudlet,tong2016hierarchical,liu2018edge}).
Dai \textit{et al.}~\cite{dai2018joint} develops an efficient computation offloading algorithm while considering computation resource allocation and transmission power allocation at the same time.
Tan \textit{et al.}~\cite{tan2017online} proposes an algorithm that aims to reduce latency by reallocating compute tasks among cloud servers. 
Similarly, work by Jia \textit{et al.}~\cite{jia2016cloudlet} shows that the latency between edge devices and the cloud could be minimized through a load-balancing framework. 
Tong \textit{et al.}~\cite{tong2016hierarchical} aims to reduce the task execution time by using a task dispatch algorithm and hierarchical cloud architecture to dynamically allocate tasks to different edge servers. 
Wang \textit{et al.}~\cite{wang2020joint} and Chen \textit{et al.}~\cite{chen2021energy} propose to jointly optimize server assignment and resource management in MEC-based MAR system.
However, these works focus mainly on reducing service latency or energy and do not take into account other task success measures in their objective function.
In contrast to the above-mentioned works, Liu \textit{et al.}~\cite{liu2018edge} and Wang \textit{et al.}~\cite{wang2020user} propose to tackle the trade-off between network latency, energy, and video object detection accuracy for mobile augmented reality systems.
In their papers, they overlook the context of the Metaverse and proceed to directly formulate a convex optimization problem. 
Subsequently, they employ standard convex optimization techniques to address it. 
This approach, while methodologically sound, fails to take into account the unique aspects and challenges presented by the Metaverse-based MEC system.
Our paper distinguishes itself from previous task-offloading studies by focusing on MEC in the context of the Metaverse.
We introduce a new optimization function tailored to the Metaverse's unique user utility model, providing new insights into virtual economic paradigms. 
By addressing these novel aspects, our work advances MEC research for the rapidly evolving Metaverse, enhancing user utility and resource management in this emerging paradigm.
% However, as our focus is on improving the mobile device-Metaverse P2E scheme, our work considers a user-association framework with the aim of minimizing the overall computation, communication service latency and maximizing the total in-game earnings of the Metaverse P2E users.

\vspace{-4pt}\section{System model} \label{sec-System}\vspace{-1pt}

% \begin{figure}[tb]
% \centering
% \includegraphics[width=0.4\textwidth]{figure/intro.pdf} \vspace{-5pt} 
% \caption{The system model of the Metaverse-based MAR system.} \label{fig_intro}\vspace{-10pt}
% \end{figure}

In this section, we give a description of the system model for Metaverse-based MEC system. The system model consists of transmission delay, computation delay and profit MAR users can make in the Metaverse \textit{play-to-earn} (P2E) applications. 

% Figure \ref{fig_intro} illustrates the system model. 
We consider a model which has $K$ Metaverse users, where $\mathcal{K}$ denotes the set of users, and $N$ MSP edge servers, where $\mathcal{N}$ denotes the set of MSP edge servers. 
The Metaverse users communicate with MSP edge servers via wireless channels. The proposed system model can be generalized to any other MAR applications in Metaverse.

\vspace{-2pt}\subsection{Communication Model}\vspace{-1pt}

Without loss of generality, we consider \textit{Orthogonal Frequency Division Multiple Access} (OFDMA) communication technique in the communication model.
Let $B$ be the total bandwidth of the system.
Assume the bandwidth is equally allocated to all the Metaverse users. 
According to Shannon's formula, the achievable uplink transmission rate of user $k$ and MSP edge server $n$ is
\vspace{-7pt}\begin{align}
    R_{k,n}^u = \frac{B}{K} \log_2 (1+\frac{g_k p_k}{B \sigma^2 / K}),\\[-21pt]\nonumber
\end{align}
where $g_k$ and $p_k$ are the channel gain and transmit power of user $k$ and $\sigma^2$ denotes the noise power density.

The communication delay between MSP edge servers and MAR users is determined by the transmitted data size and the transmission rate of the wireless channels. In the Metaverse P2E applications, the MAR users have to upload data to the MSP edge servers for computing and then download the processed data afterwards.

Denote $D_k^u$ and $D_k^d$ as the size of uplink data and downlink data for MAR user $k$ respectively.
The uplink data contains the users' posture, orientation, and position information.
The size of the uplink data depends on factors like the tracking technology being used, the complexity of the user's movements, and the frequency at which the data is sampled.
How to design the tracking technology is beyond the scope of this paper.
Thus, we consider the size of uplink data as constant without optimizing it.
The size of downlink data is decided by the resolution of the downloaded video.
Let $s_k$ denote the resolution (the number of pixels) of the video downloaded by user $k$.
The required resolution $s_k$ for the Metaverse users is denoted by $s_{min} \leq s_k \leq s_{max}$, where $s_{min}$ denotes 720p (1280$\times$720 pixels, also known as HD) and $s_{max}$ represents 8k (7680$\times$4320 pixels, popularly known as Full Ultra HD) Metaverse service, respectively.
Although the video resolution can only take on discrete values in practice, we treat it as a continuous variable during optimization.
In real-world applications, the continuous values can then be rounded or discretized to the nearest valid value supported by the display or video codec being used.
The Metaverse frame is composed of two images, one for each eye, and each pixel in these images is represented by 24 bits.
Thus, the relationship between the downlink data size $D_k$ and the resolution of the video $s_k$ can be given as $D_k^d = \frac{24\times2\times s_k}{Com_k}$,
where $Com_k$ is the compression ratio of user $k$.
Let $R_{k,n}^d$ be the downlink transmission rate between MAR user $k$ and MSP edge server $n$. 
Then, once an MAR user is allocated to an MSP edge server, the uplink and downlink wireless transmission delay experienced by the MAR user $k$ is modeled as
\vspace{-1pt}\begin{align}
    L_k^u = \frac{D_k^u}{R_{k,n}^u}, L_k^d = \frac{D_k^d}{R_{k,n}^d},
\end{align}
where $L_k^u$ is the uplink latency and $L_k^d$ is the downlink latency.
Accordingly, to transmit data within a time duration $L_k^u$, the energy consumption of user $k$ is expressed as
\vspace{-5pt}\begin{align}
    E_k = p_k L_k^u = \frac{p_kD_k^u}{R_{k,n}^u}.
\end{align}

As the MSP edge servers are always plugged in, we only focus on the user-side resource and experience.
Thus, the energy consumption of the MSP edge servers is not considered in this paper.

\vspace{-4pt}\subsection{Computation Model}\vspace{-1pt}
When MAR user $k$ uploads its data to MSP edge server $n$, the computation process at MSP edge server $n$ incurs computation latency. This latency is related to the computation resources available of MSP edge server $n$. 
Let $a_{k,n} \in \{0,1\}$ represent the MSP edge server allocation indicator, where $a_{k,n}=1$ means MAR user $k$ associates with MSP edge server $n$ for computing. 
Otherwise, $a_{k,n}=0$. We assume that each MAR user can only associate with one MSP edge server at a time. The user-server association rule can be given as:
\begin{align}
    \left\{
        \begin{aligned}
            &a_{k,n} \in \{0,1\},\forall k\in \mathcal{K},\forall n\in \mathcal{N}, \\
            &\sum_{n \in \mathcal{N}} a_{k,n} = 1,\forall k\in \mathcal{K}.
        \end{aligned}
    \right.
\end{align}
For an MSP edge server $n$, the number of MAR users it associates with is $\sum_{k \in \mathcal {K}} a_{k,n}$.
Following the computation model in~\cite{liu2018edge}, we assume the computation resource of an MSP edge server is equally distributed among the MAR users who are associated with the MSP edge server. 
Therefore, the computing resource a Metaverse user can utilize is given as
\vspace{-2pt}\begin{align}
c_n = \frac{f_n}{\sum_{k \in \mathcal{K}}a_{k,n}}, \label{C_n}
\end{align}
where $f_n$ is the computational recourses of MSP edge server $n$. 
The computational latency experienced by a user is determined by the level of computational complexity required to perform their task, as well as the extent of computational resources available on the servers~\cite{sardellitti2015joint}.
We define $\lambda_u$ and $\lambda_d$ as the computational complexity per unit of uplink data and downlink data requires. The computation overhead depends on the size of uplink and downlink data. Then, the computation delay experienced by MAR user $k$ is formulated as
\vspace{-1pt}\begin{align}
    L_k^p = \sum_{n \in \mathcal{N}} ( a_{k,n} \frac{\lambda_u D_k^u + \lambda_d D_k^d}{c_n} ). 
\end{align}
\vspace{-15pt}\subsection{Metaverse \textit{Play-to-Earn} Model}\vspace{-3pt}
Metaverse is a digital environment where users generate content and interact with others. Users are able to interact with this digital virtual environment and also produce content that affects the local digital environment. 
With the idea of open economy and financial rewards, in P2E model, users add value by playing and spending time in the P2E gaming system. Motivated by the definition of evaluation function in~\cite{xiong2020contract}, in this paper we define the following resolution-impacted earning function:
%evaluation function:
\vspace{-5pt}\begin{align}
    % g(D_k^d)= P \cdot \left[ 1 - \exp\left( - \frac{c_d D_k^d + c_u D_k^u}{\tau_k}\right)\right],
    g_k(s_k)= \tau_k h_k(s_k),\label{g_(x)}\\[-16pt]\nonumber
\end{align}
where $\tau_k>0$ is a measure of different users' earning abilities and $h_k(s_k)$ indicates the players' experience.
The earning function will be further investigated in Section \ref{sec-function}.
\vspace{-1pt}\begin{assumption}\label{assumption:earning_function}
$g_k(s_k)$ is an increasing, twice differentiable and concave function with respect to $s_k>0$; i.e., $\frac{\mathrm{d}g_k(s_k)}{\mathrm{d} s_k}>0, \frac{\mathrm{d}^2g_k(s_k)}{\mathrm{d} s_k^2}<0, \forall s_k > 0$.
\end{assumption}\vspace{-2pt}
The specific expression of $g_k(s_k)$ will be investigated in Section \ref{sec-function}.
% Table \ref{table_notations} summarizes the main notations.

% \begin{table}[tp]
% \caption{Table of Notations.} \label{table_notations} \vspace{-5pt}
% \resizebox{\linewidth}{!}{$
% \begin{tabular}{cc}
% \toprule
% \textbf{Notation}& \textbf{Description}\\
% \hline
% $K, \mathcal{K}$& The number and the set of the Metaverse users\\
% $N, \mathcal{N}$&  The number and the set of the MSP edge servers\\
% $R_{k,n}^u$ & The uplink transmission rate of user $k$ and MSP edge server $n$\\
% $R_{k,n}^d$ & The downlink transmission rate of user $k$ and MSP edge server $n$\\
% $B$ & Total bandwidth of the system\\
% $g_k$ & Channel gain between edge server and user $k$ \\
% $p_k$ & Transmit power of user $k$\\
% $\sigma^2$ & Noise power density\\
% $D_k^u, D_k^d$ & The size of uplink data and downlink data of MAR user $k$\\
% $s_k$ & The resolution of the downlink video for user $k$\\
% $Com_k$ & The compression ratio of user $k$\\
% $L_k^u, L_k^d$ & Uplink and downlink wireless transmission delay\\
% $E_k$ & The energy consumption for use $k$ to transmit the uplink data\\
% $f_n$ & The computation recourse of MSP edge server $n$\\
% $a_{k,n}$ & MSP edge server allocation indicator\\
% $c_n$ & Computing resource available to each user on MSP server $n$\\
% $\lambda_u, \lambda_d$ & Computational complexity per data unit in uplink and downlink\\
% $L_k^p$ & The computation delay experienced by MAR user $k$\\
% $\tau_k$ & The measure of user $k$'s earning ability\\
% $\omega$ & Weight parameter\\
% $\eta_e,\eta_l$ & Normalization factors\\
% \bottomrule
% \end{tabular}
% $}
% \end{table}

\vspace{-5pt}\section{Problem Formulation and Analysis} \label{sec-analysis}\vspace{-5pt}

In this section, we formulate the problem as a multi-objective optimization problem. Then the analysis and solutions to this problem are provided.
\begin{figure*}[htb]
\vspace{-10pt}
\centering
\includegraphics[width=0.88\textwidth]{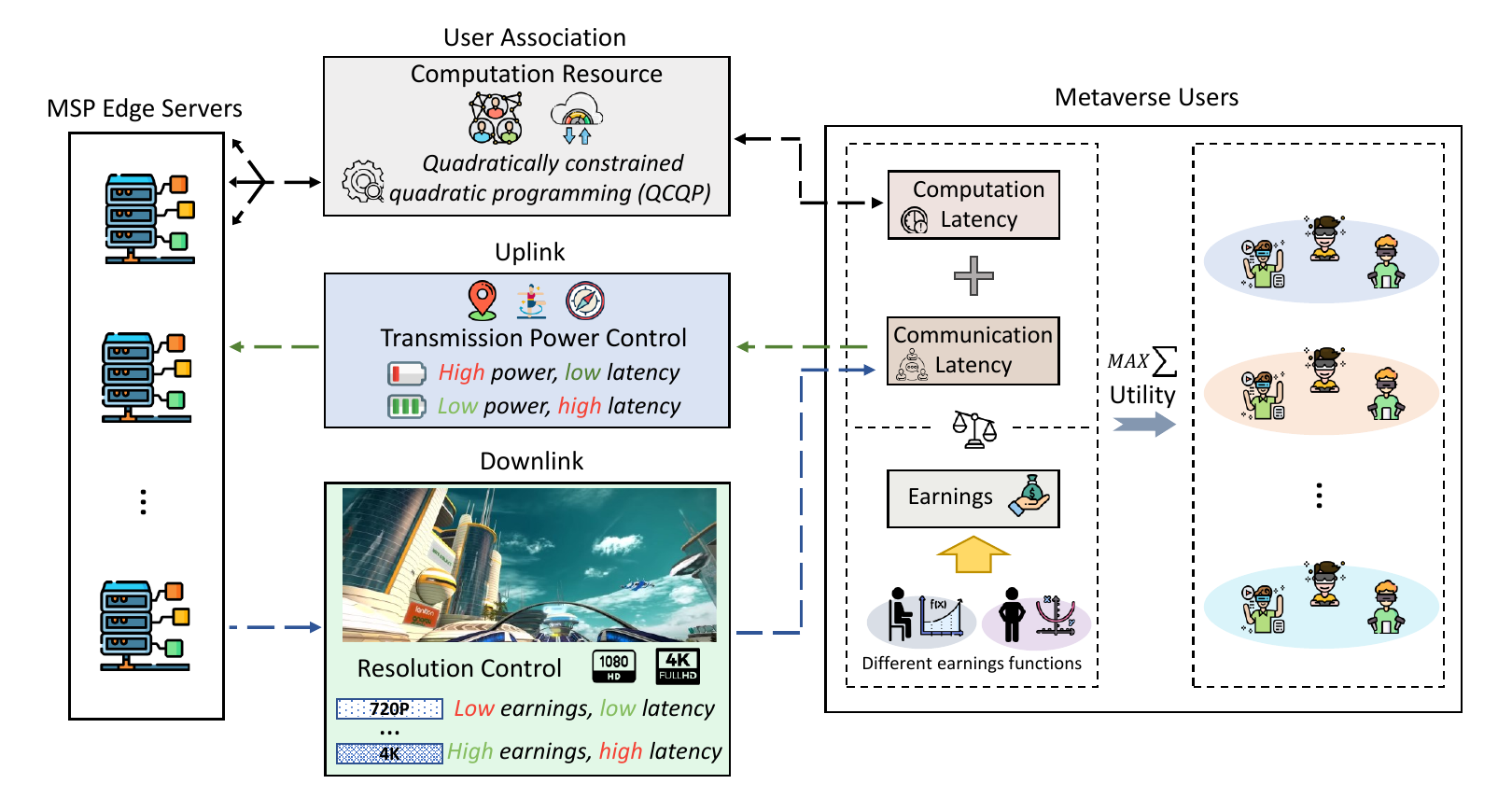} \vspace{-5pt}
\vspace{-3pt}\caption{An overview of the system framework and the proposed optimization problem in the Metaverse-based MAR system. The multi-objective problem aims to maximize the overall user utility through the management of the user-to-edge association, uplink transmission, and downlink transmission.} \label{fig_system}\vspace{-23pt}
\end{figure*}
\vspace{-5pt}\subsection{Problem Formulation}\vspace{-3pt}
Based on the communication model, computation model and Metaverse P2E model, the utility of MAR user $k$ can be given as weighted sum~\cite{marler2010weighted} of latency and earnings:
\vspace{-5pt}\begin{align}\label{user_utility}
    U_k:= \eta_e \cdot g_k(s_k) -  \eta_l \cdot\omega (L_k^u + L_k^d + L_k^p),\\[-20pt]\nonumber
\end{align}
where $\omega$ is a positive weight parameter indicating the preference between delay and earnings. $\eta_e$ and $\eta_l$ are normalization factors. By setting different $\omega$, the trade-off between earnings and user-perceived latency can be controlled.
In the Metaverse-based MAR system, we aim to minimize the latency of all MAR users and maximize their earnings. Thus, the multi-objective optimization problem is formulated as:

\begin{subequations}\label{pro:max_u}
\begin{align}
    \mathscr{P}_1: \max_{\{s_k\},\{p_k\},\{a_{k,n}\}} &\sum_{k \in \mathcal{K} }U_k \tag{\ref{pro:max_u}}\\[-3pt]
    \text{s.t.~}~& s_{min} \leq s_k \leq s_{max},\forall k\in \mathcal{K}, \label{con:resolution_1}\\
    &a_{k,n} \in \{0,1\},\forall k\in \mathcal{K},\forall n\in \mathcal{N},\label{con:UA_1}\\
    ~& \sum_{n \in \mathcal{N}} a_{k,n} = 1,\forall k\in \mathcal{K}, \label{con:UA_2}\\[-3pt]
    &E_k \leq E_k^{max}, \forall k\in \mathcal{K}, \label{con:energy_1}\\
    &0 \leq p_k \leq p_{max}, \forall k\in \mathcal{K}. \label{con:power}\\[-18pt]\nonumber
\end{align}
\end{subequations}
In this problem, the objective function is the sum of users' utilities, which is a weighted sum of latency and earnings.
Constraint (\ref{con:resolution_1}) confines the selection of video resolution in the download process. Constraints (\ref{con:UA_1}) and (\ref{con:UA_2}) guarantee that one MAR user is assigned to one MSP edge server.
Constraint (\ref{con:energy_1}) requires that the energy of uplink transmission for user $k$ does not exceed an energy budget $E_k^{max}$. 
Constraint (\ref{con:power}) limits the maximum transmit power of the user $k$ to $p_{max}$.
The $E_k^{max}$ depends on the battery capacity or energy limit configured by the user.
Figure~\ref{fig_system} presents an overview of the system framework and the proposed optimization problem.

Problem $\mathscr{P}_1$ is considered to be a mixed-integer nonlinear programming with nonlinearity in the objective as well as continuous and integer variables in the constraints. It combines the combinatorial difficulty of optimizing over discrete variable sets with the challenges of handling nonlinear function. Thus, it is difficult to solve~\cite{belotti2013mixed}.
% To solve the problem, we relax the discrete constraints of $a_{k,n}$ to take a continuous form. The relaxed variable is denoted as $a_{k,n}$ such that $0 \le a_{k,n} \le 1$.
% After obtaining the continuous solutions, it will be rounded back to the closest discrete values.
To determine the maximum value of $\sum_{k \in \mathcal{K} }U_k$, it's equivalent to identifying the minimum value of $-\sum_{k \in \mathcal{K} }U_k$. 
Then problem $\mathscr{P}_1$ can be rewritten in the following form:
% To solve the problem, first we decompose it into to sub-problems. The first sub-problem aims to optimize the download data size $\{D_k^d\}$ with fixed user-to-server association
\vspace{-5pt}\begin{subequations}\label{pro_min_u}
\begin{align}
    \mathscr{P}_2: \min_{\{s_k\},\{p_k\},\{a_{k,n}\}}& F = -\sum_{k \in \mathcal{K} }U_k \tag{\ref{pro_min_u}}\\
    \text{s.t.}~ &\text{(\ref{con:resolution_1}),~(\ref{con:UA_1}),~(\ref{con:UA_2}),~(\ref{con:energy_1}),~(\ref{con:power}).} \nonumber\\[-21pt]\nonumber
\end{align}
\end{subequations}
% {\color{red}The average earnings of all users can be specified as:

% \begin{align}
% \sum_{k \in K}\zeta\left(D^d_k\right)
% \end{align}
% }
% The Hessian of $e^{-\frac{D^d_k}{\tau_k}}$ is positive definite while the \textit{play-to-earn} function $g(D^d_k)$ = $P \cdot (1 - e^{-\frac{D^d_k}{\tau_k}})$ is a linear? function. Therefore, the sum of average earnings of all uses $\sum_{k \in \mathcal{K}}\zeta(D^d_k)$ is a convex function.\\
% \\
% The total latency can be specified as:

% \begin{align}
% L = \sum_{n \in \mathcal{N}}\sum_{k \in \mathcal{K}} \left(\frac{D_k^d}{R_k^d}+\frac{D_k^u}{R_k^u} + \frac{ \lambda D_k^u}{f_n} \sum_{m \in  \mathcal  {K}}  a_{m,n}\right)
% \end{align}
In the subsequent section, we begin by identifying the optimal transmit power $\{p_k\}$  tailored for the multi-objective optimization problem.
Following this, the user-to-server association and the downlink video resolution that maximizes the overall user utility are determined.
\vspace{-3pt}\subsection{Optimal Transmit Power}
The proposition presented below provides a systematic approach for determining the optimal transmit power tailored for each Metaverse user $k$.
\begin{prop}\label{prop1}
Given the user-to-server association $\{a_{k,n}\}$ and downlink video resolution $\{s_k\}$, the optimal transmit power $p_k$ of user $k$ is given by
\vspace{-5pt}\begin{align}
    p_k^* = \min \{p_{max},\hat{p_k}\}, \label{solution_pk}\\[-21pt]\nonumber
\end{align}
where
\vspace{-5pt}\begin{align}
    \hat{p_k} = -\frac{BE_k^{max}}{\ln2\cdot D_k^u K}W(-\frac{\ln2\cdot D_k^u \sigma^2}{E_k^{max}g_k}e^{-\frac{\ln2\cdot D_k^u \sigma^2}{E_k^{max}g_k}})-\frac{B \sigma^2}{g_kK}. \label{p_Lambert}\\[-21pt]\nonumber
\end{align}
Here, $W(\cdot)$ denotes Lambert $W$ function.
% where $\hat{p_k}$ satisfies the equality: $\frac{\hat{p_k}D_k^uK}{B\log_2 (1+\frac{g_k \hat{p_k}}{B \sigma^2 / K})}=E_k^{max}$.
\end{prop}
\begin{proof}
See Appendix \ref{Appendix:1}.
\end{proof}
Building upon Proposition \ref{prop1}, the optimal transmit power is determined.
Next, we will tackle the downlink video resolution and the user-to-server association problem.
\subsection{Alternating Optimization}
To solve problem $\mathscr{P}_2$, a sub-optimal alternative algorithm based on alternating optimization is proposed in this subsection. 
First, we introduce the following two lemmas.
\begin{lem}\label{lemma2}
The problem $\mathscr{P}_2$ with fixed transmit power and user-to-server association is a convex problem with respect to resolution $\{s_k\}$.
\end{lem}
\begin{proof}
For any feasible $s_{min} \leq s_k \leq s_{max}$,
% \begin{align}
%       \frac{\partial^2 F}{\partial D^d_i \partial D^d_j} = \left\{
%     \begin{aligned}
%         &\frac{\eta_e P\cdot e^{-\frac{D_i^d}{\tau_i}}}{\tau_i^2},&i=j,\\
%         & \quad\quad0, &\text{otherwise.}
%     \end{aligned}
%     \right.
% \end{align}
\vspace{-1pt}\begin{align}
       \frac{\partial^2 F}{\partial s_i \partial s_j} = \left\{
    \begin{aligned}
        &-\eta_e \tau_k \frac{\mathrm{d}^2h_i(s_i)}{\mathrm{d} s_i^2},&i=j,\\
        & \quad0, &\text{otherwise.} \\
    \end{aligned}
    \right.\\[-18pt]\nonumber
\end{align}

% \begin{align}
%     \frac{\partial^2 F}{\partial D^d_i \partial D^d_j} = 
%          \omega\cdot\left(\frac{\partial^2(\psi)}{\partial D^d_i \partial D^d_j}\sum_{n \in \mathcal{N}}\frac{a_{i,n}}{f_n}\sum_{m \in \mathcal{K}}a_{m,n}\right)\\ - \frac{\partial^2 \zeta}{\partial D^d_i \partial D^d_j}
%     \end{align}
According to Assumption \ref{assumption:earning_function}, $g_k(s_k)$ is a twice differentiable and concave function. Since $\tau_k>0$, $\frac{\mathrm{d}^2h_i(s_i)}{\mathrm{d} s_i^2}<0$.
$\eta_e$ is a positive value, thus, $-\eta_e \tau_k \frac{\mathrm{d}^2h_i(s_i)}{\mathrm{d} s_i^2} > 0$. Thus, the Hessian matrix $\mathbf{H}_{ij}^1=\frac{\partial^2 F}{\partial s_i \partial s_j}$ is symmetric and positive definite. Constraint (\ref{con:resolution_1}) is linear. Therefore, problem $\mathscr{P}_2$ is convex with respect to $\{s_k\}$ under fixed user-to-server association.
\end{proof}
With Lemma \ref{lemma2}, the optimal resolution $\{s_k\}$ can be obtained by utilizing convex optimization solvers such as CVX~\cite{grant2008cvx}.
\vspace{-3pt}\begin{lem}\label{lemma1}
The problem $\mathscr{P}_2$ with fixed transmit power and resolution selection is a non-convex problem with respect to $\{a_{k,n}\}$.
\end{lem}
\begin{proof}
For any $a_{i,j},a_{m,n}, \forall i,m \in \mathcal{K}, j,n \in \mathcal{N}$,
    % \begin{align}
    %     \frac{\partial^2 F}{\partial a_{i,j}\partial a_{m,n}} = \left\{
    %     \begin{aligned}
    %         &\frac{2 \omega \lambda (D_i^u+D_i^d)}{f_j},&i = j~ \text{and}~m = n,\\
    %         & ~~~~~~~~~0, &\text{otherwise.}
    %     \end{aligned}
    %     \right.
    % \end{align}
% The Hessian matrix given by $\mathbf{H}_{ij,mn} = \frac{\partial^2 F}{\partial a_{i,j}\partial a_{m,n}}$ is symmetric
% and positive definite~\cite{liu2018edge}.
% The constraints $C_3$  and $C_4$ in (\ref{pro:max_u2}) are linear while $C_1$ is not related to the user-to-server association. Thus, the problem formulated is convex with respect to $\{a_{k,n}\}$.
\vspace{-5pt}\begin{align}
    &\frac{\partial^2 F}{\partial a_{i,j}\partial a_{m,n}} = \nonumber\\&\left\{
    \begin{aligned}
        &\frac{\eta_l \omega \left[\lambda_u (D_i^u+D_m^u)+ \lambda_d(D_i^d+D_m^d)\right]}{f_j},&j = n,\\
        &\quad\quad\quad\quad\quad\quad\quad~~0,&\!\!\!\!\!\!\!\!\text{otherwise.}
    \end{aligned}
    \right.\\[-18pt]\nonumber
\end{align}
The Hessian matrix given by $\mathbf{H}_{ij,mn}^2 = \frac{\partial^2 F}{\partial a_{i,j}\partial a_{m,n}}$ is not symmetric
and positive definite~\cite{liu2018edge}. 
Thus, the mixed-integer non-linear programming problem is a non-convex problem, which is considered to be NP-hard generally~\cite{sheng2016intelligent}.
The quadratic objective function also makes it difficult to tackle.
\end{proof}\vspace{-2pt}
Next, we propose using alternating optimization to obtain the user-to-server association $\{a_{k,n}\}$ and the resolution $\{s_k\}$. Specifically, we optimize $\{a_{k,n}\}$ and $\{s_k\}$ iteratively.
Without loss of equivalence, the integer constraint (\ref{con:UA_1}) can be reformulated as:
\vspace{-5pt}\begin{align}
    &0 \leq a_{k,n} \leq 1, \forall k\in \mathcal{K},\forall n\in \mathcal{N}, \label{con:UA_3}\\
    &\sum_{k \in \mathcal{K}} \sum_{n \in \mathcal{N}} a_{k,n}(1-a_{k,n}) \leq 0. \label{con:UA_4}\\[-15pt]\nonumber
\end{align}
Thus, problem $\mathscr{P}_2$ with fixed resolution selection and transmit power is given as:
\vspace{-8pt}\begin{align}
    \mathscr{P}_3: &\min_{\{a_{k,n}\}} - \sum_{k \in \mathcal{K} }U_k,\label{obj_P_3} \\
    \text{s.t.}~&\text{(\ref{con:UA_2}),~(\ref{con:UA_3}),~(\ref{con:UA_4}).}\\[-21pt]\nonumber
\end{align}
According to Lemma \ref{lemma1}, problem $\mathscr{P}_3$ is non-convex. In order to solve the non-convex problem, we first transform it into a quadratically constrained quadratic programming (QCQP) problem. Then, the SDR approach is applied to obtain an approximate solution to this problem.
\subsubsection{QCQP Transformation}\label{subsection:QCQP_transformation}
We introduce a new column vector $\boldsymbol{a}$ with $M = K\times N$ elements:
\vspace{-5pt}\begin{align}
    \boldsymbol{a} \!= \!\left[a_{1,1},a_{1,2},\ldots,a_{1,n},a_{2,1},\ldots,a_{2,n},a_{k,1},a_{k,2},\ldots,a_{k,n}\right]^T.\\[-21pt]\nonumber
\end{align}
With fixed downlink video resolution $\{s_k\}$, the earnings and the transmission delay in the objective function (\ref{pro_min_u}) are constants. We remove the constant part, and then the objective function (\ref{obj_P_3}) in problem $\mathscr{P}_3$ is simplified into:
\begin{align}
    \min_{\{a_{k,n}\}}& \eta_l \omega \sum_{k \in \mathcal{K} } \sum_{n \in \mathcal{N}} ( a_{k,n} \frac{\lambda_u D_k^u + \lambda_d D_k^d}{c_n} ). \label{pro:a}\\[-21pt]\nonumber
\end{align}
Substitute (\ref{C_n}) into (\ref{pro:a}), we can get:
\vspace{-5pt}\begin{align}
    \min_{\{a_{k,n}\}}& \eta_l \omega \!\!\sum_{k \in \mathcal{K} } \!\sum_{n \in \mathcal{N}} ( a_{k,n} \frac{(\lambda_u D_k^u \!+\! \lambda_d D_k^d)\sum_{k_2 \in \mathcal{K}}a_{k_2,n}}{f_n} ).\label{pro:a2}\\[-21pt]\nonumber
\end{align}
Then, we replace user-to-server association $\{a_{k,n}\}$ with vector $\boldsymbol{a}$ in (\ref{pro:a2}). Before this, we define matrices $\mathbf{P}$ and $\mathbf{Q}$ as 
\begin{align}
        \mathbf{Q}&= \left.\left[
            \begin{array}{ccccc}
              \boldsymbol{e}_1 & \mathbf{0}&   \cdots &  \mathbf{0} \\
                \mathbf{0}  & \boldsymbol{e}_1&   \cdots &  \mathbf{0}  \\
              \vdots &  \vdots & \ddots &  \vdots  \\
                \mathbf{0} &  \mathbf{0}&   \cdots & \boldsymbol{e}_1  \\
             \end{array}
         \right]\right\}k,
\\
         \mathbf{P}&= \left.\left[
        \begin{array}{ccccc}
          \mathbf{J}_1    & \mathbf{J}_2    &  \cdots & \mathbf{J}_k    \\
          \mathbf{J}_1    & \mathbf{J}_2    &  \cdots & \mathbf{J}_k    \\
          \vdots &  \vdots & \ddots & \vdots \\
          \mathbf{J}_1    & \mathbf{J}_2    &  \cdots & \mathbf{J}_k    \\
         \end{array}
     \right]\right\}k,
    %  Q &= 
    %     \left.\left[ {\begin{array}{*{20}{ccccc}}
    %         {\overbrace {\begin{array}{*{20}{cccc}}
    %         1&1&\cdots&1\\
    %         0&0&\cdots&0\\
    %         \vdots&\vdots&\ddots&\vdots\\
    %         0&0&\cdots&0\\
    %         \end{array}}^k}&{\overbrace {\begin{array}{*{20}{cccc}}
    %         1&1&\cdots&1\\
    %         0&0&\cdots&0\\
    %         \vdots&\vdots&\ddots&\vdots\\
    %         0&0&\cdots&0\\
    %         \end{array}}^k}&{\overbrace {\begin{array}{*{20}{cccc}}
    %         1&1&\cdots&1\\
    %         0&0&\cdots&0\\
    %         \vdots&\vdots&\ddots&\vdots\\
    %         0&0&\cdots&0\\
    %         \end{array}}^k}
    %         \end{array}} \right]\right\}k.
\end{align}
where $\boldsymbol{e}_1$ and $\mathbf{J}_k$ are given by
\vspace{-5pt}\begin{small}
\begin{align}
    &\boldsymbol{e}_1~~ =~~ {\overbrace {\left[1,1,\ldots,1\right]}^n},\\
      &\mathbf{J}_k= \nonumber\\&\left[\! 
        \begin{array}{ccccc}
        \frac{\lambda_u D_k^u + \lambda_d D_k^d}{f_1} & 0 & 0 & \cdots  & 0   \\
        0 & \frac{\lambda_u D_k^u + \lambda_d D_k^d}{f_2} & 0 & \cdots  & 0   \\
        % 0 & 0 & \frac{\lambda_u D_k^u + \lambda_d D_k^d}{f_3} & \cdots  & 0   \\
        \vdots &  \vdots & \vdots  &\ddots & \vdots  \\
        0 & 0 & 0 & \cdots  & \frac{\lambda_u D_k^u + \lambda_d D_k^d}{f_n}\\
         \end{array}
     \! \right]_{n \times n}.
\end{align}
\end{small}
Let $\boldsymbol{q}_j$ denote the $j$-th row of $\mathbf{Q}$. 
Then, problem $\mathscr{P}_3$ can be rewritten as:
% Then, constraint (\ref{con:UA_2}) is equivalent to:
% \begin{align}
%     \boldsymbol{q}_j^T \boldsymbol{a} = 1, \forall j \in \{1,2,\cdots,k\}. \label{con:c5}
% \end{align}
% With $\mathbf{P}$ and $\boldsymbol{q}_j$, problem $\mathscr{P}_3$ can be rewritten as:
\vspace{-5pt}\begin{align}
    \mathscr{P}_{3.1}: \min_{\boldsymbol{a}} ~& \eta_l \omega \boldsymbol{a}^T \mathbf{P}\boldsymbol{a},\\
    \text{s.t.~}~
    &\boldsymbol{q}_j \boldsymbol{a} = 1, \forall j \in \{1,2,\ldots,k\}, \label{con:c5}\\
    &\text{(\ref{con:UA_3}),~(\ref{con:UA_4}).}\nonumber\\[-21pt]\nonumber
\end{align}
Note that the constraint (\ref{con:c5}) is equivalent to (\ref{con:UA_2}).
The constraint $a_{k,n} \leq 1$ can be removed in (\ref{con:UA_3}), since constraint (\ref{con:c5}) ensures that $\sum_{n \in \mathcal{N}} a_{k,n} = 1$. Based on this, problem $\mathscr{P}_{3.1}$ is further transformed as an equivalent problem:
\vspace{-5pt}\begin{subequations}\label{pro:3.2}
\begin{align}
    \mathscr{P}_{3.2}: \min_{\boldsymbol{a}} ~& \eta_l \omega \boldsymbol{a}^T \mathbf{P}\boldsymbol{a},\tag{\ref{pro:3.2}}\\
    \text{s.t.~}~
    &a_{k,n} \geq 0, \forall k\in \mathcal{K},\forall n\in \mathcal{N}, \label{con:c6}\\
    &\text{(\ref{con:UA_4}),~(\ref{con:c5}).} \nonumber\\[-21pt]\nonumber
\end{align}
\end{subequations}
Defining $\boldsymbol{b} \triangleq$
\begin{math}
    \left[
        \begin{matrix}
        \boldsymbol{a} \\
         1
        \end{matrix}
    \right]
\end{math}
and $\mathbf{P}_1 \triangleq$
\begin{math}
    \left[
        \begin{matrix}
        \mathbf{P}_{M \times M} & \mathbf{0}_{M \times 1} \\
         \mathbf{0}_{1 \times M} & 0
        \end{matrix}
    \right]
\end{math}, problem $\mathscr{P}_{3.2}$ is then homogenized to:
\begin{subequations}\label{pro:4}
\begin{align}
    \mathscr{P}_4:&\min_{\boldsymbol{b}} ~ \eta_l \omega \boldsymbol{b}^T \mathbf{P}_1\boldsymbol{b},\tag{\ref{pro:4}}\\
    \text{s.t.~}~ &b_i \geq 0, \forall i \in \{1,2,\ldots,M\},\label{con:c7}\\
    &\boldsymbol{b}^T \mathbf{Y} \boldsymbol{b} \leq 0,\label{con:UA_5}\\
    &\boldsymbol{b}^T \mathbf{G}_j\boldsymbol{b} = 1, \forall j \in \{1,2,\ldots,k\},\label{con:c8} \\
    &b_{M+1}=1,\label{con:c9}
\end{align}
\end{subequations}
where $b_{M+1}$ is the $(M+1)$-th element of vector $\boldsymbol{b}$.
$\mathbf{Y}$ and $\mathbf{G}_j$ are given by
\begin{align}
    \mathbf{Y} &=         \left[
        \begin{array}{cc}
        -\mathbf{I}_{M \times M} & \frac{1}{2}\boldsymbol{e}_2^T \\
         \frac{1}{2}\boldsymbol{e}_2 & 0
        \end{array}
    \right]_{(M+1) \times (M+1)},\\
    \mathbf{G}_j &= 
        \left[
        \begin{array}{cc}
        \mathbf{0}_{M \times M} & \frac{1}{2}\boldsymbol{q}_j^T \\
         \frac{1}{2}\boldsymbol{q}_j & 0
        \end{array}
    \right]_{(M+1) \times (M+1)},
\end{align}
where $\mathbf{I}_{M \times M}$ is a $M\times M$ identity matrix and $\boldsymbol{e}_2 = {\overbrace {\left[1,1,\ldots,1\right]}^M}$.
In problem $\mathscr{P}_4$, 
constraints (\ref{con:c7}), (\ref{con:UA_5}) and (\ref{con:c8}) are equivalent to constraints (\ref{con:c6}), (\ref{con:UA_4}) and (\ref{con:c5}), respectively. 
% \begin{scriptsize}
% \[
%     \begin{bmatrix}
%         \frac{D_1^u+D_1^d}{f_1} & 0 & \cdots  & 0 & \frac{D_2^u+D_2^d}{f_1} & \cdots & D_k^u+D_k^d  \\
%         \vdots &  &  & \ddots  &  & \cdots & \\
%         D_1^u+D_1^d & D_1^u+D_1^d  & \cdots  & D_2^u+D_2^d & D_2^u+D_2^d & \cdots & D_k^u+D_k^d      \\
%         D_1^u+D_1^d & D_1^u+D_1^d & \cdots  & D_2^u+D_2^d & D_2^u+D_2^d & \cdots & D_k^u+D_k^d \\
%     \end{bmatrix}_{kn \times kn}
% \]
% \end{scriptsize}
\subsubsection{Semidefinite Relaxation (SDR)}\label{subsection:SDR}
Problem $\mathscr{P}_4$ is still non-convex and difficult to solve. To find a solution to problem $\mathscr{P}_4$, we use the semidefinite programming (SDP) approach, which relaxes the problem into an SDR problem.

First, we define $\mathbf{B} \triangleq \boldsymbol{b}\boldsymbol{b}^T$, then $\boldsymbol{b}^T \mathbf{P}_1\boldsymbol{b} = \Tr (\mathbf{P}_1 \mathbf{B})$, $\boldsymbol{b}^T \mathbf{Y} \boldsymbol{b} = \Tr(\mathbf{Y}\mathbf{B})$, $\boldsymbol{b}^T \mathbf{G}_i\boldsymbol{b} = \Tr(\mathbf{G}_i \mathbf{B})$ and rank$(\mathbf{B})=1$. Problem $\mathscr{P}_4$ is equivalent to
\begin{subequations}\label{pro:4.1}
\begin{align}
    \mathscr{P}_{5}:&\min_{\mathbf{B}}~\Tr(\mathbf{P}_1 \mathbf{B})\tag{\ref{pro:4.1}}\\
    \text{s.t.}~&\mathbf{B}_{i,j} \geq 0, \forall i,j \in \{1,2,\ldots,M\},\label{con:c10}\\
    &\Tr(\mathbf{Y}\mathbf{B}) \leq 0,\label{con:UA_6}\\
    &\Tr(\mathbf{G}_j \mathbf{B}) = 1,\forall j \in \{1,2,\ldots,k\},\label{con:c11}\\
    &\mathbf{B}_{M+1,M+1} = 1,\label{con:c12}\\
    &\mathbf{B} \succeq 0, \label{con:c13}\\
    &\text{rank}(\mathbf{B})=1,\label{con:c14}
\end{align}
\end{subequations}
where $\mathbf{B}_{i,j}$ is the element at row $i$ and column $j$ of matrix $\mathbf{B}$. 
Constraint (\ref{con:c10}) ensures that all the elements are non-negative. Thus, all the elements in $\boldsymbol{b}$ have the same sign. 
Together with constraint (\ref{con:c11}), it makes sure all the elements in $\boldsymbol{b}$ are non-negative, which is equivalent to constraints (\ref{con:c7}). 
Constraints (\ref{con:UA_6}), (\ref{con:c11}) and (\ref{con:c12}) are equivalent to constraints (\ref{con:UA_5}), (\ref{con:c8}) and (\ref{con:c9}), respectively. 
In problem $\mathscr{P}_{5}$, the only difficulty is the non-convex constraint (\ref{con:c14}), whereas other constraints as well as the objective function are convex. 
By dropping constraint (\ref{con:c14}), we obtain an SDR of problem $\mathscr{P}_{5}$:
\begin{align}
    &\mathscr{P}_{5.1}:\min_{\mathbf{B}}~\Tr(\mathbf{P}_1 \mathbf{B})\\
    &\text{s.t.}~\text{(\ref{con:c10}),~(\ref{con:UA_6}),~(\ref{con:c11}),~(\ref{con:c12}),~(\ref{con:c13}).} \nonumber
\end{align}
Until now, Problem $\mathscr{P}_{5.1}$ becomes a convex optimization problem, which can be solved in polynomial time through convex optimization solvers such as CVX~\cite{grant2008cvx}. 

\subsubsection{Randomization} \label{subsection:Randomization}
The semidefinite relaxation in problem $\mathscr{P}_{5.1}$ provides a lower bound on the optimal value of problem $\mathscr{P}_{5}$.
However, it is still not clear how to compute a good feasible solution to $\mathbf{B}$.
Next, we employ a randomization technique to obtain an approximate solution that is not guaranteed to be exact but is computationally efficient.

Let $\mathbf{B}^*$ be the optimal solution to problem $\mathscr{P}_{5.1}$.
If $\text{rank}(\mathbf{B}^*)=1$, then the optimal solution to problem
$\mathscr{P}_4$ can be constructed by $\mathbf{B}^* = \boldsymbol{b}^*\boldsymbol{b}^{*T}$. 
\begin{algorithm}[t] 
\caption{User utility maximization algorithm} \label{algorithm1}
    \textbf{Input:} Weight $\omega$ and the initial downlink video resolution $\{s_k\}$.\\
    \textbf{Output:} The transmit power $\{p_k\}$, the user-to-server association $\{a_{i,j}\}$ and the downlink video resolution $\{s_k\}$.\vspace{-2pt}
    \begin{algorithmic}[1]
        \STATE {Find the optimal transmit power $p_k$ based on (\ref{solution_pk}).}
        \REPEAT
            \STATE {$\{a_{k,n}\} \gets$ solve problem $\mathscr{P}_2$ under fixed $\{s_k\}$ according to section \ref{subsection:QCQP_transformation}, section \ref{subsection:SDR} and section \ref{subsection:Randomization};}\vspace{-12pt}
            \STATE {$\{s_k\} \gets$ solve problem $\mathscr{P}_2$ under fixed $\{a_{k,n}\}$ according to lemma \ref{lemma2};}
        \UNTIL{convergence}
        % \FOR{$k \in \mathcal{K}$}
        %     \STATE {$n' = \argmax_{j \in \mathcal{N}}a_{k,j}$;}
        %     \STATE {$a_{k,n'} = 1$;}
        %     \FOR{$n \in \mathcal{N}$ and $n \neq n'$}
        %         \STATE {$a_{k,n} = 0$;}
        %     \ENDFOR
        % \ENDFOR
    %\STATE {$\{D_k^d\} \gets$ solve problem $\mathscr{P}_2$ under given $\{a_{k,n}\}$;}
    \RETURN $\{p_k\}, \{a_{i,j}\}, \{s_k\}$
    \end{algorithmic}
\end{algorithm}
On the other hand, if the rank of $\mathbf{B}^*$ is not 1, it indicates the optimal objective value of problem $\mathscr{P}_{5.1}$ only serves as a lower bound of problem $\mathscr{P}_{5}$.
In this case, we present a Gaussian randomization-based approach for obtaining an approximate solution.
First, we apply eigen-decomposition to $\mathbf{B}^*$ as $\mathbf{B}^*=\mathbf{U\Lambda U}^H$, where $\mathbf{U}$ a unitary matrix and $\mathbf{\Lambda}$ is a diagonal matrix. After that, Gaussian Randomization is applied to obtain a quasi-optimal solution. 
Specifically, we generate $\boldsymbol{r} \in \mathbb{R}^{(M+1)\times1}$ such where the $j$-th element $r _j\sim N(0,1), j \in \{1,2,\ldots,M+1\}$ and $N(0,1)$ here denotes standard Gaussian distribution with a mean of zero and standard deviation of 1. 
The Gaussian randomization generates $l$ random samples. Let the superscript $(l)$ denote the index of a random sample. Then we obtain candidate vector by $\bar{\boldsymbol{b}}^{(l)}=\mathbf{U\Lambda^{1/2}}\boldsymbol{r}^{(l)}$
The final estimate is the one that corresponds to the minimum objective value among all the candidate vectors. Finally, we obtain $\boldsymbol{a}$ by removing the $(M+1)$-th element of $\boldsymbol{b}$. Then, $\{a_{k,n}\}$ can be constructed by $\boldsymbol{a}$.

After the user-to-server association $\{a_{k,n}\}$ is obtained, the resolution $\{s_k\}$ is optimized based on the fixed $\{a_{k,n}\}$. The alternating optimization algorithm optimizes $\{s_k\}$ and $\{a_{k,n}\}$ step by step until the objective function is converged. 
The proposed algorithm is introduced in Algorithm \ref{algorithm1}.

SDP can be handled by most convex optimization toolboxes with an interior point algorithm. 
Thus, the worst case complexity of solving problem $\mathscr{P}_{5.1}$ is $\mathcal{O}((M+1)^{\frac{9}{2}}\log(\frac{1}{\epsilon}))$ given a solution accuracy $\epsilon$~\cite{luo2010semidefinite}.
% After obtaining the optimal $\{a_{k,n}\}$, we convert the continuous value to integers as shown in line 6 to line 12 in Algorithm \ref{algorithm1}.

% {\color{red}
% With Lemma \ref{lemma1} and Lemma \ref{lemma2}, the problem can be solved by block coordinate descent method (BCD)~\cite{grippo2000convergence} by fixing one variable and updating the other variable. Firstly, we initialize all the variables. Then, with fixed downlink data size $D_k^d$, the user-server association $a_{k,n}$ is obtained by solving problem $\mathscr{P}_2$. 
% After that, the downlink data size $D_k^d$ is optimized based on the fixed $a_{k,n}$. The algorithm optimizes $D_k^d$ and $a_{k,n}$ step by step until the objective function is converged. 
% The proposed algorithm is introduced in Algorithm \ref{algorithm1}. After obtaining the optimal $a_{k,n}$, we convert the continuous value to integers as shown in line 6 to line 12 in Algorithm \ref{algorithm1}.
% The BCD algorithm is proved to converge to the optimal value in a sub-linear rate~\cite{beck2013convergence}.
% }

\section{Evaluation of the \textit{Play-to-Earn} Earning Function} \label{sec-function}
\begin{figure*}[htb]
\vspace{-10pt}
\centering
\includegraphics[width=0.85\textwidth]{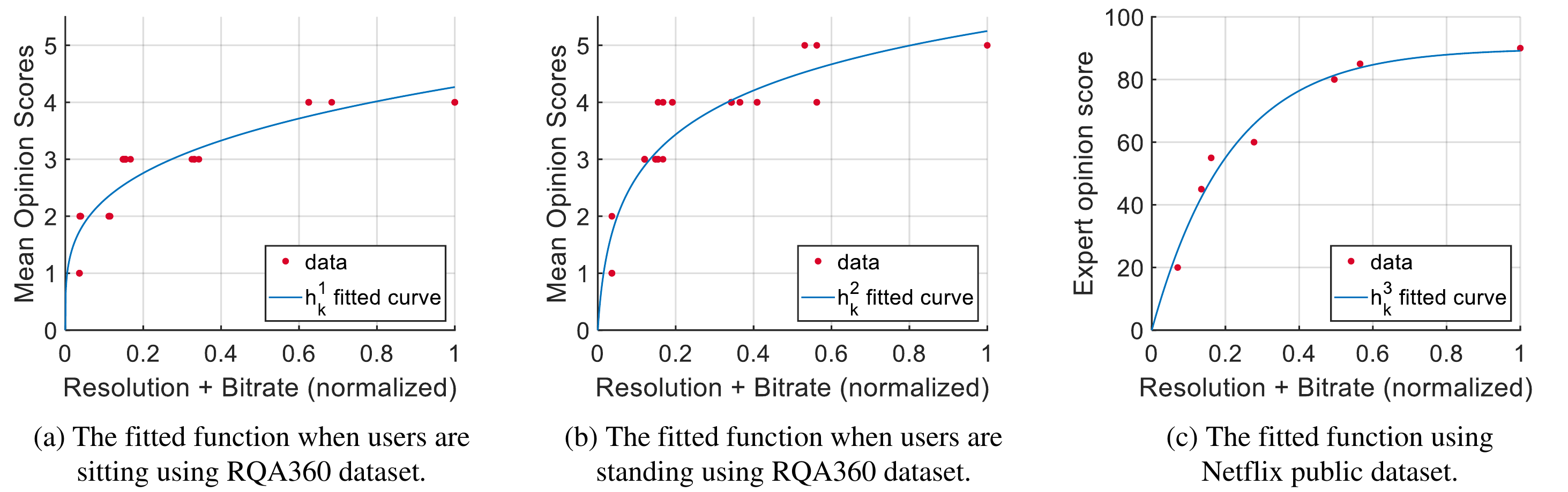} 
\vspace{-5pt}\caption{Three different types of earning functions.}
\label{fig_fitted}\vspace{-17pt}
\end{figure*}
% \begin{figure}[htb]
% \vspace{-10pt}
% \centering
% \includegraphics[width=5.5cm]{figure/fitted_P002_seated_rb.pdf} 
% \vspace{-5pt}\caption{The fitted earning function when the users are sitting.}
% \label{fig_fitted_p002_seated}\vspace{-5pt}
% \end{figure}

% \begin{figure}[htb]
% \vspace{-10pt}
% \centering
% \includegraphics[width=5.5cm]{figure/fitted_P002_stand_rb.pdf} 
% \vspace{-5pt}\caption{The fitted earning function when the users are standing.}
% \label{fig_fitted_p002_stand}\vspace{-5pt}
% \end{figure}
In this section, we investigate the earning functions $g_k(s_k)$ in (\ref{g_(x)}) by analyzing the RQA360 dataset~\cite{elwardy2021opinion} and Netflix public dataset~\cite{Netflix_dataset}.

\textbf{RQA360 dataset.} The RQA360 dataset is a dataset on subjective quality assessment of 360$^\circ$ videos.
The participants view 360$^\circ$ video stimuli on an HTC Vive device equipped with an integrated eye tracker.
Then, the participants are tasked with assessing the quality of a set of videos by assigning scores based on their opinions in two separate sessions. 
One session requires participants to view the videos while seated, while the other session requires participants to view the videos while standing.
Then the scores are measured by Mean Opinion Scores (MOS), ranging from 0 to 5.

\textbf{Netflix public dataset.} It is a dataset provided by Netflix publicly available to the community.
Each 2D video is evaluated by a professional assessor who assigns an Expert Opinion Score (EOS) ranging from 0 to 100. 

Intuitively, there should be a proportionate relationship between the players' profitability and MOS/EOS in P2E games for the following reasons:
\begin{itemize}[leftmargin=*]
\item \textbf{There is a positive correlation between profitability and player engagement.}
In Metaverse games, MOS/EOS is a measure of user satisfaction, which is strongly correlated with user engagement in the game.
A higher MOS/EOS usually indicates smoother, higher-definition graphics. It increases the visibility of the player, thereby improving the engagement of players.
A higher engagement level means more time spent on the game, which in turn can lead to higher profitability for the player.

\item \textbf{Higher MOS/EOS indicates better game design, bringing higher profitability.}
A high MOS/EOS indicates that the game is well-designed, with intuitive mechanics, enjoyable gameplay, and an engaging environment. 
These factors can lead to higher player retention and a better reputation for the game, which can attract more players and ultimately result in higher profitability for the players.

\item \textbf{MOS/EOS reflects the willingness of players to spend.}
A high MOS/EOS suggests that players are more enjoyable and willing to spend time and money on the game.
This willingness to spend can ultimately turn into higher profitability for the players, especially in P2E games where players are rewarded for their time and effort.
% \item \textbf{Higher MOS indicates better skill of player:}
% In many P2E games, the rewards are tied to player performance. A high MOS may indicate that a player is more skilled at the game, which can lead to higher rewards and profitability.
\end{itemize}
Therefore, we define $g_k(s_k) = \tau_k h_k(s_k)$ where $h_k(s_k)$ indicates the players' experience and $\tau_k$ is a measure of player $k$'s earning ability. In the P2E games, the rewards are tied to player performance. A player who is more skilled at the game can earn higher rewards and profitability.
% Since the user-experienced latency (i.e., communication delay $L_k^u, L_k^d$ and computation delay $L_k^p$) is already incorporated in the user utility (\ref{user_utility}), in the earning functions we only aim at the players' income.
Next, we study how the resolution $s_k$ impacts users' MOS/EOS.
In order to address the limited data available for videos with the same bitrate, we consider both resolution and bitrate as variables by combining them.
% In contrast, in the case of low-quality graphics, the visual experience is impaired, which may pose specific challenges such as difficulty perceiving the game environment. Players in this situation are more likely to perform poorly in games, thus yielding fewer tokens.
% $\tau_k$ is a measure of user $k$'s earning ability. 
% It describes how size of downlink data influences the profit of different users. For instance, a user with excellent gaming skills has lower $\tau_k$ and can earn more under the limited downlink data. 
% $P$ is a scaling factor which defines the maximum profit a user can make. $g(\cdot)$ is a non-decreasing and concave function with decreasing marginal earnings. This indicates the fact that end users' marginal earnings are dwindling as they download more data, which is in line with reality.

We use three different commonly used utility functions to estimate the users' MOS when they are in different postures:
\begin{align}
    h_k^1(s_k) &= \alpha_1 (s_k+R_{k,n}^d)^{\beta_1},\label{earning_func1}\\h_k^2(s_k) &= \alpha_2 \ln[1+ \beta_2 (s_k+R_{k,n}^d)], \label{earning_func2}\\
    h_k^3(s_k) &= \alpha_3 [1-e^{-\beta_3 (s_k+R_{k,n}^d)}],\label{earning_func3}
\end{align}\vspace{-3pt}
where $\alpha_1, \alpha_2,\alpha_3, \beta_2,\beta_3 \ge 0$ and $0\le\beta_1\le1\vspace{+2pt}$.
We normalize resolution $s_k$ and downlink bitrate $R_{k,n}$ to the range of 0 to 0.5, respectively.
Thus, the sum of $s_k$ and $R_{k,n}$ is normalized to 0 to 1.
The function $h_k^1$ was introduced by Jiang \textit{et al.}~\cite{jiang2005max} to reflect the system's utility in wireless networks.
Yang \textit{et al.}~\cite{yang2015incentive} proposed function $h_k^2$ to capture the crowdsourcer's dwindling return in crowdsensing system.
Liu \textit{et al.}
~\cite{liu2018edge} proposed function $h_k^3$ to describe the relationship between the analytics accuracy and the video frame resolution in MEC.
$h_k^1$, $h_k^2$ and $h_k^3$ are all increasing, twice differentiable and concave.
Thus, they follow Assumption \ref{assumption:earning_function}.
It is worth noting that the proposed algorithm and results of this paper can be extended to any other earning functions satisfying Assumption \ref{assumption:earning_function}.
\begin{table}[tp]
\centering
\vspace{+5pt}
\caption{Parameter of the earning functions.}
\label{table_earning_functions}
\vspace{-5pt}
\setlength{\tabcolsep}{3mm}{
\begin{tabular}{cccc}
\toprule
\textbf{Parameter}& \textbf{Value}&\textbf{Parameter}& \textbf{Value}\\
\hline
$\alpha_1$&  4.268 & $\beta_1$ & 0.2714 \\\hline
$\alpha_2$&  1.159  & $\beta_2$ & 91.92 \\\hline
$\alpha_3$&  89.95  & $\beta_3$ & 4.732 \\
\bottomrule
\end{tabular}
}\vspace{-15pt}
\end{table}
Figure \ref{fig_fitted} (a) and Figure \ref{fig_fitted} (b) present the fitted earning functions $h_k^1$ and $h_k^2$ when users are in different postures using RQA360 dataset.
Figure \ref{fig_fitted} (c) shows the fitted function $h_k^3$ using Netflix dataset.
It can be observed from the data that the MOS grows more rapidly when users stand as opposed to when they sit.
Thus, we use $h_k^1$ and $h_k^2$ to fit the different situations.
The difference can be attributed to the fact that standing provides more space for users to perform physical movements, which enables them to achieve better in-game experience.
For 2D videos, the user experience grows much more slowly.
Thus, we use $h_k^3$ to fit Netflix dataset.
Table~\ref{table_earning_functions} summarizes the fitting parameters of the three different functions.

% \begin{figure}[htb]
% \centering
% \begin{minipage}[t]{0.23\textwidth}
% \centering
% \includegraphics[width=4.5cm]{figure/fitted_P002_seated.pdf} 
% \caption{The fitted curve when the users are sitting.}
% \label{fig_fitted_p002_seated}
% \end{minipage}
% \begin{minipage}[t]{0.23\textwidth}
% \centering
% \includegraphics[width=4.5cm]{figure/fitted_P002_stand.pdf} 
% \caption{The fitted curve when the users are standing.}
% \label{fig_fitted_p002_stand}
% \end{minipage}
% \end{figure}

\vspace{-5pt}\section{Simulation Results} \label{sec-Simulation}
In this section, we evaluate the performance of the proposed method. We simulate a network with 20 MSP edge servers and 100 Metaverse users.
For the earning function, we choose them randomly with normalized MOS/EOS.
The parameters regarding the earning functions are presented in Table \ref{table_earning_functions}.
The computational recourses $f_n$ of MSP edge servers are chosen randomly from 1 to 5 TFLOPS.
The computational complexity per unit of data $\lambda_u$ is uniformly distributed between 1 to 10 KFLOPS while $\lambda_d$ is uniformly distributed between 1 to 100 KFLOPS per pixel.
The compression ratios $Com_k$ is selected from the uniform distribution $[300,600]$~\cite{bastug2017toward}.
The path loss is modeled as $128.1 + 37.6\log(\texttt{distance})$ and Gaussian noise power is $\sigma^2=-134\text{dBm}$.
The maximum transmit power $p_{max}$ is 0.2 W.
The downlink transmission rates are uniformly distributed between 10 to 20 Mbps.
% The MAR users are deployed in a square area of size 300 m $\times$ 300 m with MSP edge servers randomly placed betweem them.
The number of Gaussian randomization $l$ is set to be 1000.
Firstly, we validate our method by comparing the latency with respect to varying weight parameter $\omega$. 
The value of $\omega$ describes the weight of the influence of latency on the system. Next, we evaluate the average earnings of users under different minimum video resolutions. 
The amount of average earnings is scaled to $[0,1]$. We compare the proposed method with three kinds of algorithms summarized as follows:

% \textbf{Latency-optimized only:} The latency-optimized only method aims at minimizing the total latency of the system only, regardless of the amount of earnings that MAR users may earn.

% \textbf{Earnings-optimized only:} The earnings-optimized only method optimizes the downlink video resolution to maximize the earnings without considering the latency. It adopts random user-to-server association. 

% \textbf{Baseline:} The baseline method picks random downlink video resolution and random user-to-server association.

\begin{itemize}[leftmargin=*]
    \item \textbf{Optimal-latency:} The optimal-latency method aims at minimizing the total latency of the system, regardless of the amount of earnings that MAR users may earn.
    \item \textbf{Optimal-earnings:} The optimal-earnings method optimizes the downlink video resolution to maximize the earnings without considering the latency. It adopts random user-to-server association. 
    \item \textbf{Baseline:} The baseline method picks random downlink video resolution and random user-to-server association.
\end{itemize}
\vspace{-5pt}\subsection{Comparison under different weight parameter}
In this sub-section, we evaluate the latency and earnings given different weight parameter $\omega$. The value of $\omega$ indicates how latency impacts the system. 
A higher value of $\omega$ puts more emphasis on latency. In the evaluation of latency, we set the weight parameter $\omega$ from 0.5 to 5 as shown in Figure \ref{fig_latency_w}. 
It can be seen that the proposed method achieves lower latency with the increase of $\omega$, which means the system emphasizes more on the latency requirement. In the evaluation of earnings, we set the weight parameter $\omega$ from 0.5 to 5. Figure \ref{fig_earnings_w} shows the impact of $\omega$ on the earnings. With the increase of $\omega$, our method gets lower earnings.

% \begin{figure}
% \begin{minipage}[b]{0.25\textwidth}
% \includegraphics[width=1\textwidth]{figure/latency_w_mini.pdf}\\
% % \subcaption{A subfigure}
% \end{minipage}%
% \begin{minipage}[b]{0.25\textwidth}
% \includegraphics[width=1\textwidth]{figure/earnings_w_mini.pdf}\\
% % \subcaption{Another subfigure}
% \end{minipage}%
% \caption{Several figures}
% \end{figure}

\begin{figure}[tb]
\centering
\includegraphics[width=0.31\textwidth]{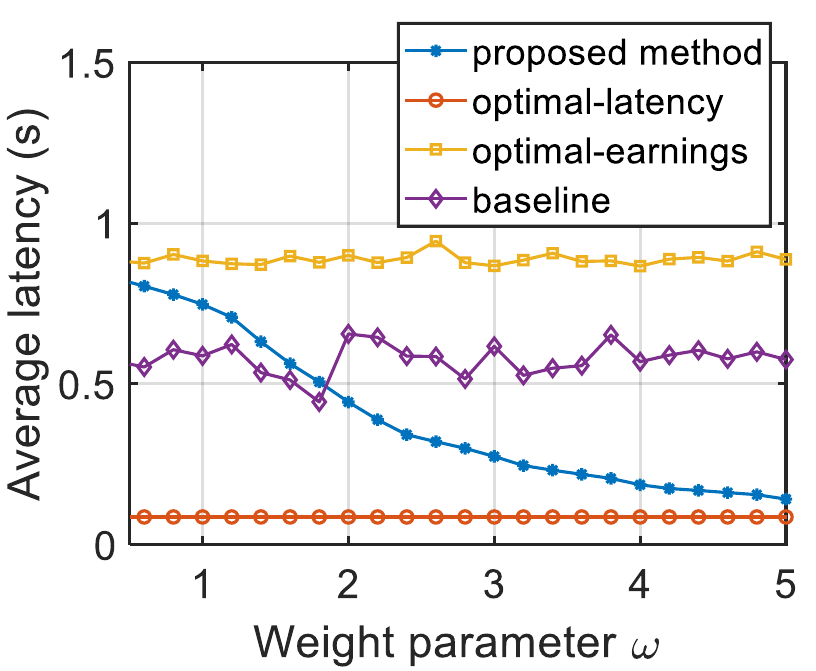} 
\vspace{-5pt}\caption{The latency under different weight $\omega$.}
\label{fig_latency_w}\vspace{-18pt}
\end{figure}

\begin{figure}[tb]
\centering
\includegraphics[width=0.31\textwidth]{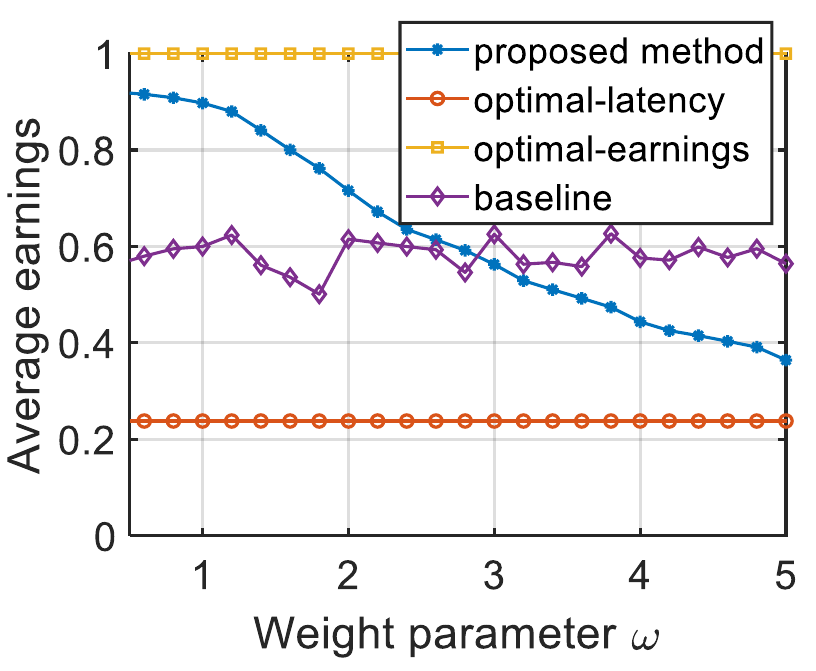} 
\vspace{-6pt}\caption{The earnings under different weight $\omega$.}
\label{fig_earnings_w}\vspace{-15pt}
\end{figure}

The lower earnings obtained by our method is due to its compromise of earnings for latency. The optimal-latency method keeps low latency all the time since it only aims at minimizing the latency. 
Thus, it chooses minimum downlink video resolution all the time, and also obtains low earnings which can be seen in Figure \ref{fig_earnings_w}. 
For the optimal-earnings method, as illustrated in Figure \ref{fig_latency_w}, the latency is relatively high because of the higher downlink video resolution it chooses. At the same time, it has the highest earnings among all the methods. 
For the baseline, the user-to-server association is random. Thus the latency fluctuates between the optimal-earnings method and the optimal-latency method. 
It can be observed from Figure \ref{fig_latency_w} that our method has similar latency with baseline when $\omega = 1.5$ while achieving higher earnings than baseline, which can be derived from Figure \ref{fig_earnings_w}. 
The experiments indicate that when choosing appropriate $\omega$ (e.g., $\omega = 2.75$), the proposed method maintains similar earnings but reduces about 40\% system latency compared with the baseline. 
When increasing $\omega$ (e.g., $\omega \geq 5$), the proposed method achieves almost similar latency performance with optimal-latency method but improves the earnings substantially. 

\begin{figure}[tb]
\centering
\includegraphics[width=0.31\textwidth]{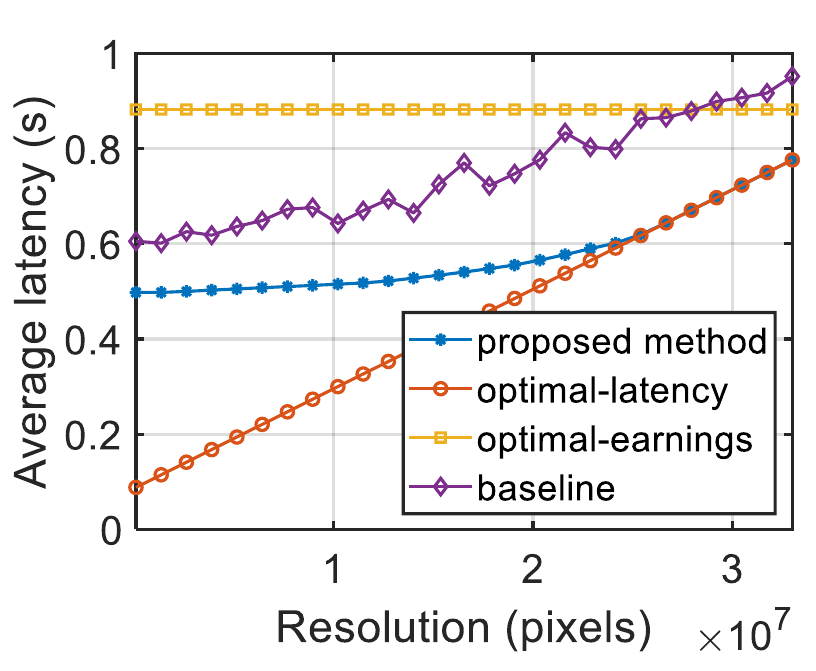}
\vspace{-6pt}\caption{The latency under different minimum video solution $s_{min}$.}
\label{fig_latency_D}\vspace{-15pt}
\end{figure}

\vspace{-5pt}\subsection{Comparison under different minimum downlink video resolution}
In Metaverse-based MAR system, minimum downlink video resolution requirements ensure that Metaverse users are eligible to participate in the P2E games. Next, we evaluate the latency and earnings given different $s_{min}$. 
A higher downlink video resolution will result in more earnings for the users, while incurring more transmission and computation latency, thus degrading user utility.
In the simulation, we set the minimum downlink video resolution $s_{min}$ from $1280\times720$ pixels (720p) to $6400\times4800$ pixels (nearly 8k). Figure \ref{fig_latency_D} shows the latency with respect to different $s_{min}$. From Figure \ref{fig_latency_D} and Figure \ref{fig_earnings_D}, it can be found that the latency and earnings of proposed method almost remain unchanged when the minimum downlink video resolution requirement is less than $1.5\times10^7$ pixels. 
This is because the downlink video resolution for most Metaverse users is optimized. 
The optimal-latency method aims at minimizing the overall latency regardless of players' earnings. 
Therefore, it picks the minimum downlink video resolution all the time. It can be seen from Figure \ref{fig_latency_D} that the optimal-latency method always has the lowest latency. 
For the optimal-earnings method as shown in Figure \ref{fig_earnings_D}, it always achieves the highest earnings, while maintaining a high latency. For the baseline, it picks random user-to-server association and downlink video resolution.

\begin{figure}[tb]
\centering
\includegraphics[width=0.31\textwidth]{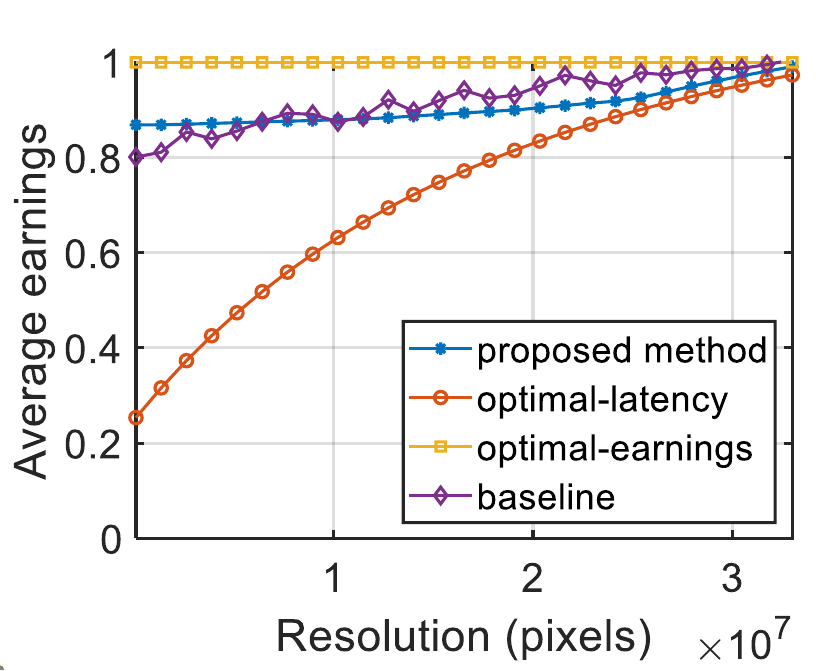} 
\vspace{-6pt}\caption{The earnings under different minimum video solution $s_{min}$.}
\label{fig_earnings_D}\vspace{-14pt}
\end{figure}

Figure \ref{fig_latency_D} and Figure \ref{fig_earnings_D} show that the proposed method obtains a low latency at the cost of earnings. However, compared with the baseline, the proposed method reduces about 30\% latency while maintaining similar earnings when the minimum downlink video resolution is $2\times10^7$ pixels. 
It indicates that our method is able to reduce the latency while bringing more profitability to users, thus improving user utility.

\vspace{-3pt}\subsection{The impact of number of the Metaverse users}
We next assess the influence of the number of users in Figure~\ref{fig_bar}.
As the number of users increases, the workload of MSP edge servers also increases, leading to higher computational latency. 
The proposed algorithm attains the lowest latency compared to the other methods. 
The performance gap between the proposed approach and the three other algorithms demonstrates the gains acquired from the optimal server allocation.
\begin{figure}[tb]
\centering
\includegraphics[width=0.43\textwidth]{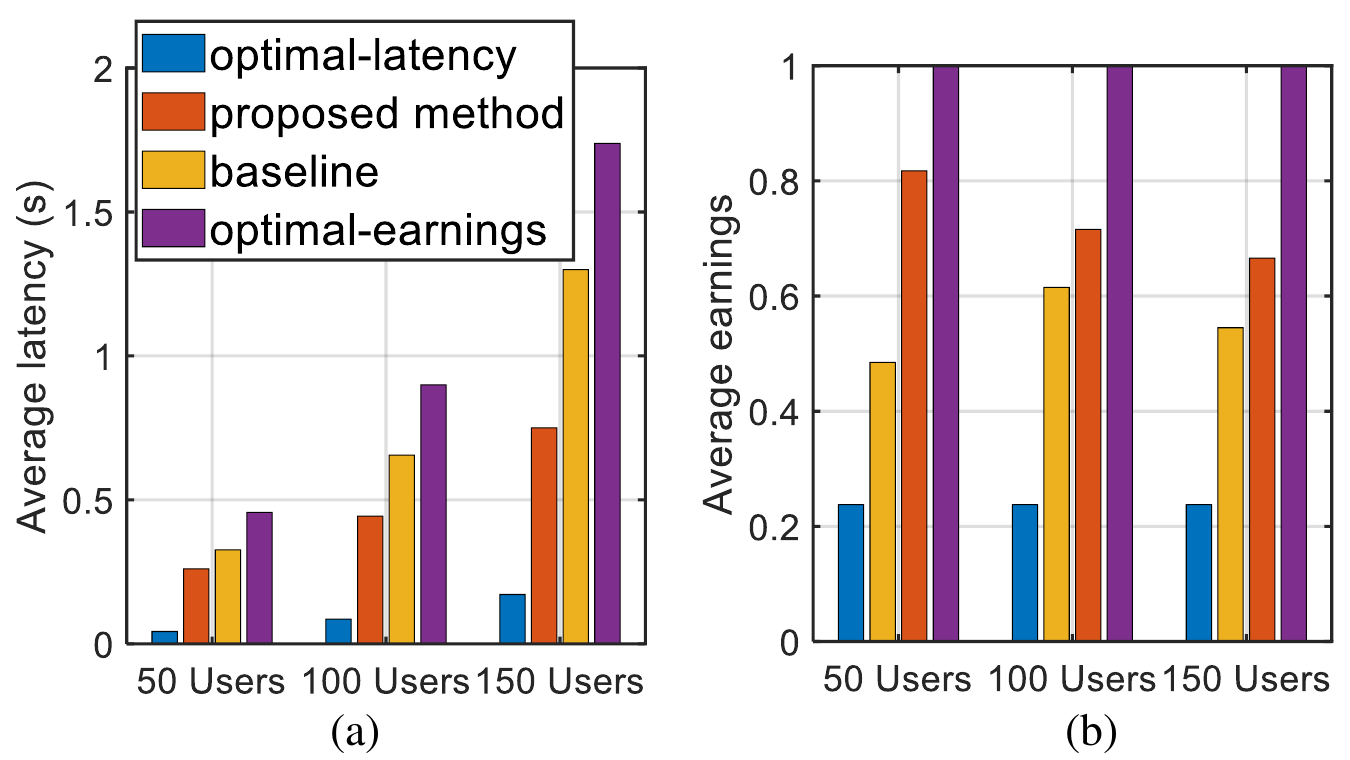} 
\vspace{-5pt}\caption{The impact of the number of the Metaverse users.}
\label{fig_bar}\vspace{-17pt}
\end{figure}
Compared to the baseline algorithm, the proposed algorithm exhibits a significant improvement in service latency, up to 30\% reduction, with a slight increase of 10\% in average earnings when the user count is 100. As the number of users increases, the MSP edge servers become overloaded, causing an increase in computational load and service latency. To address this issue, the proposed algorithm adopts a proactive approach of reducing the downlink video resolution to ensure low service latency, which, in turn, leads to a slight reduction in earnings. The proposed algorithm outperforms other algorithms in terms of both latency and earnings performance.

\section{Conclusion} \label{sec-Conclusion}\vspace{-1pt}

In this paper, we proposed the Metaverse system over MEC. We explained the profitability of playing games in the virtual world. To solve the trade-off between latency and potential token earnings, we formulated a multi-objective optimization problem. 
The semidefinite relaxation (SDR) approach is applied to get an approximate solution by transforming the problem into a quadratically constrained quadratic programming (QCQP) problem.
Then the optimal user-to-server association and downlink video resolution are found through alternating optimization.
The simulation results demonstrated our method can effectively balance the latency and players' profitability, thus improving the performance of the Metaverse systems over MEC.

\appendix
\section{Proof of Proposition 1}\label{Appendix:1}
First, the solution (\ref{p_Lambert}) to $\hat{p_k}$ can be derived from the following equality:
\vspace{-8pt}\begin{align}
    \frac{\hat{p_k}D_k^uK}{B\log_2 (1+\frac{g_k \hat{p_k}}{B \sigma^2 / K})}=E_k^{max}. \label{p_hat} \\[-20pt]\nonumber
\end{align}
To prove this, we first prove $y(p_k)\vspace{+3pt}$ is an increasing function of $p_k$ where $y(p_k) = \frac{p_kD_k^uK}{B\log_2 (1+\frac{g_k p_k}{B \sigma^2 / K})}$.
The first derivative of $y(p_k)$ with respect to $p_k$ is
\begin{small}
\vspace{-5pt}\begin{align}
\frac{\mathrm{d}y(p_k)}{\mathrm{d} p_k} \!=\! \dfrac{\ln2 \!\cdot\! D_k^uK\!\left[\left(Kg_kp_k\!+\!B{\sigma}^2\right)\ln\!\left(\frac{Kg_kp_k}{B{\sigma}^2}\!+\!1\right)\!-\!Kg_kp_k\right]}{B\left(Kg_kp_k\!+\!B{\sigma}^2\right)\ln^2\left(\frac{Kg_kp_k}{B{\sigma}^2}\!+\!1\right)}.\\[-21pt]\nonumber
\end{align}
\end{small}
To investigate the sign of $\frac{\mathrm{d}y(p_k)}{\mathrm{d} p_k}$, we only need to investigate the sign of $[\left(Kg_kp_k+B{\sigma}^2\right)\ln\left(\frac{Kg_kp_k}{B{\sigma}^2}+1\right)-Kg_kp_k]$ since $D_k^u,g_k,p_k,B,\sigma$ and $K$ are all positive values.
Let $t = \frac{Kg_kp_k}{B{\sigma}^2}$ where $t > 0$, then $[\left(Kg_kp_k+B{\sigma}^2\right)\ln\left(\frac{Kg_kp_k}{B{\sigma}^2}+1\right)-Kg_kp_k]$ can be rewritten as $Kg_kp_k[(1+\frac{1}{t})\ln(t+1)-1]$.
It is obvious that $Kg_kp_k > 0$.
Next, let $V(t) = (1+\frac{1}{t})\ln(t+1)-1$.
The first derivative of $V(t)$ is $\frac{\mathrm{d}V(t)}{\mathrm{d} t} = \frac{t-\ln(t+1)}{t^2}$.
It is easy to verify that $t-\ln(t+1)>0$ when $t>0$.
Thus, $V(t)>\lim_{t \to 0}V(t)\!=\!0$ when $t>0$.
Therefore, $\frac{\mathrm{d}y(p_k)}{\mathrm{d} p_k} > 0$ and $y(p_k)$ is a monotonically increasing function of $p_k$ when $p_k>0$.

We then use contradiction. Assume that $p'_k$ ($p'_k \neq p_k^*$, $p'_k \leq  p_{max}$) is the optimal solution.
Since $p_k^* = \min \{p_{max},\hat{p_k}\}$, if $p'_k < p_k^*$, then the transmission rate will decrease, which means the objective function in problem $\mathscr{P}_2$ is not minimized.
Otherwise, if $p'_k > p_k^*$, it indicates $\hat{p_k} < p_{max}$ and $p_k^* = \hat{p_k}$.
As we have already proved that $y(p_k)$ is a monotonically increasing function of $p_k$, then $y(p'_k) > y(\hat{p_k}) = E_k^{max}$, which does not meet the constraint (\ref{con:energy_1}).
Therefore, we have $p'_k = p_k^*$. 
With this, the proof is concluded.


\end{document}